\documentclass[conference,letterpaper]{IEEEtran}
\IEEEoverridecommandlockouts
\usepackage{latexsym, amssymb, amsfonts, epsf, xcolor}
\usepackage{enumerate, color}
\usepackage{multirow,multicol}
\usepackage{bm}
\usepackage{blkarray}
\usepackage{booktabs}
\usepackage{framed}
\usepackage{setspace}
\usepackage{tikz}
\usepackage{graphicx}
\usepackage{amsthm}
\usepackage{tikz-cd}

\newtheorem{theorem}{Theorem}[section]
\newtheorem{lemma}[theorem]{Lemma}
\newtheorem{proposition}[theorem]{Proposition}
\newtheorem{corollary}[theorem]{Corollary}

\theoremstyle{definition}
\newtheorem{definition}[theorem]{Definition}

\newtheorem{remark}[theorem]{Remark}
\newtheorem{example}[theorem]{Example}

\usepackage[utf8]{inputenc} 
\usepackage[T1]{fontenc}
\usepackage{url}
\usepackage{ifthen}
\usepackage{cite}
\usepackage[cmex10]{amsmath} % Use the [cmex10] option to ensure complicance
\newcommand{\RS}{\operatorname{RS}}

\newcommand{\Per}{\operatorname{Per}}

\newcommand{\wH}{\operatorname{wt_H}}

\def\F{\mathbb{F}}
\def\Fq{{\F_q}}

\def\Z{\mathbb{Z}}
\def\cov{\mathrel{<\kern-.6em\raise.015ex\hbox{$\cdot$}}}

\newcommand{\rmv}[1]{}

\begin{document}
\title{The permutation group of Reed-Solomon codes over arbitrary points} 

% %%% Single author, or several authors with same affiliation:
% \author{%
%  \IEEEauthorblockN{Author 1 and Author 2}
% \IEEEauthorblockA{Department of Statistics and Data Science\\
%                    University 1\\
 %                   City 1\\
  %                  Email: author1@university1.edu}% }

%%% Several authors with up to three affiliations:
\author{%
    \thanks{We thank Abhay Chaudhary, Maria Chara, and Rachel Petrik for all their invaluable discussions and comments.}
    \thanks{This project started at the Rethinking Number Theory workshop in June 2022. We want to thank the organizers for the workshop, and are grateful for the supportive, collaborative research environment this workshop provided. The workshop was supported by the Number Theory Foundation, the American Institute of Mathematics, and the University of Wisconsin-Eau Claire.}
    \IEEEauthorblockN{Eduardo Camps-Moreno}
    \IEEEauthorblockA{Department of Mathematics \\
    Virginia Tech\\
    Blacksburg, VA USA\\
    e.camps@vt.edu}
    \and
  \IEEEauthorblockN{Jun Bo Lau
  \thanks{Jun Bo Lau was partially supported by the European Research Council (ERC) under the European Union’s Horizon 2020 research and innovation programme (grant agreement ISOCRYPT - No. 101020788), by the Research Council KU Leuven grant C14/24/099 and by CyberSecurity Research Flanders with reference number VR20192203.}}
  \IEEEauthorblockA{Department of \\Electrical Engineering\\
  KU Leuven\\
  Leuven, Belgium\\
  junbo.lau@kuleuven.be}
  \and
    \IEEEauthorblockN{Hiram H. L\'opez,\\~\IEEEmembership{Senior Member,~IEEE} \thanks{Hiram H. L\'{o}pez was partially supported by NSF DMS-2401558.}}
  \IEEEauthorblockA{Department of Mathematics \\
  Virginia Tech\\
  Blacksburg, VA USA\\
  hhlopez@vt.edu}
  \and
  \IEEEauthorblockN{Welington Santos}
\IEEEauthorblockA{Department of Mathematics, \\Statistics, \& Computer Science\\
University of Wisconsin-Stout \\
Menomonie, WI, USA\\
santosw@uwstout.edu}
}

\maketitle

\begin{abstract}
In this work, we prove that the permutation group of a Reed-Solomon code is given by the polynomials of degree one that leave the set of evaluation points invariant. Our results provide a straightforward proof of the well-known cases of the permutation group of the Reed-Solomon code when the set of evaluation points is the whole finite field or the multiplicative group.
\end{abstract}

\section{Introduction}
Understanding the automorphism group of a code provides valuable information on its symmetries, which can be used for efficient error correction. The optimality in error correction is only one among all the important properties of Reed-Solomon codes, making the study of their automorphism group particularly important. These automorphisms are crucial for decoding methods such as the permutation decoding technique, first introduced by J. MacWilliams in the seminal paper~\cite{MacPermutation_Decoding}, and used in, for example, \cite{9187607, 6507307, 9666868, 9729216, 9965899}. This decoding method has proven especially effective for codes with large automorphism groups.

The permutation group of Reed-Solomon codes is well studied in the literature, focusing on cases where the evaluation points are either the entire finite field or its multiplicative subgroup. Dur~\cite{DUR198769} considers a larger family of maximum distance separable codes called Cauchy codes and shows that the automorphism group of a Reed-Solomon code is isomorphic to $\left(\F^{*}_{q}\rtimes Gal(\F_q)\right)\times S_{n}$ when $k=1$ or $k=n-1$, and is isomorphic to a subgroup $G$ of $\left(\F^{*}_{q}\times GL(2,\F_q)\right)\rtimes Gal(\Fq)$ that fixes the evaluation set of the Reed-Solomon code when $2\leq k\leq n-2$. Later, Berger~\cite{10.1007/3-540-54303-1_114} employed modular algebras $A=K[G]$, where $K=GF(p^{l})\subseteq GF(p^{m})$, to represent Reed-Solomon codes as ideals. Berger's work provides an approach to understanding the automorphism group of some extended cyclic codes. In particular, Berger proved in~\cite{10.1007/3-540-54303-1_114} that the automorphism group of a Reed-Solomon code with evaluation set $\Fq$ is the affine group.

In this paper, we use elementary algebraic tools to prove straightforwardly that the permutation group of a Reed-Solomon code is determined by the affine permutations that fix the evaluation set. It is worth noting that the results in \cite{DUR198769} do not extend, at least trivially, to our setting due to technical limitations. For example, Corollary 2 in \cite{DUR198769} establishes a surjective group homomorphism from the group of linear fractional transformations that fixes the evaluation set of $C$ to the permutation group of $C$. Then, the construction yields the desired isomorphism only in the particular cases where the evaluation set is $\mathbb{F}_{q}$ or $\mathbb{F}^{*}_{q}$, and does not extend to more general evaluation sets. Similarly, the method used in \cite{10.1007/3-540-54303-1_114} relies on the code being generated by the specially chosen basis $\Theta_{k}$. Such a basis is lost if some evaluation points are removed from the evaluation set, preventing a straightforward application of the approach to our more general framework.

This paper is organized as follows. Section~\ref{preli} reviews the essential preliminaries on linear codes, the permutation group of a linear code, and Reed-Solomon codes. Section~\ref{affperm} provides a discussion of affine permutations. In Section~\ref{perRS}, we present our main results, characterizing the permutation group of a Reed-Solomon code when the dimension is not equal to $n-1$. We also provide an example to illustrate why the condition that the code dimension is $n-1$ is necessary in Theorem~\ref{24.01.19}. The paper concludes with a summary in Section~\ref{conclusion}.

%%%%%%%%%%%%%%%%%%%%%%%%%%%%%%%%%%%%%%%%%%%%%%%%%%%%%%%%%%%%%%%%%%%%%%%%%%%%%%%%%%%%%%%%%%%%%%%%%%%%%%%%%%%%%%%%%%%%%%%%%%%%%%%%%%%%%%%%%%%%%%%%%%%%%%%%%%%%%%%%%%%%%%%%%%%%%%
\section{Preliminaries}\label{preli}
Let $\F_q$ be a finite field with $q$ elements. An $[n,k,d]$ {\it linear code} over $\F_q$ is a $k$-dimensional subspace $C \subseteq \F_q^n$ with {\it minimum distance} $d := \min\{ \wH({ c}) : 0\neq  c \in C\},$ where $\wH( c)$ denotes the Hamming weight of $ c$.

The dual of $C$ with respect to the Euclidean inner product is defined by \[C^{\perp}:= \left\{  w \in \F_q^n:  w \cdot  c = 0 \ \text{ for all }  c \in C \right\},\]
where $w \cdot c$ denotes the standard Euclidean inner product.

We denote the symmetric group of degree $n$ by $S_n$. Any permutation $\pi \in S_n$ defines the map
\[\begin{array}{lccc}
& \F_{q}^n & \to & \F_{q}^n \\
& A=(a_1,\ldots, a_n) & \mapsto & \pi(A):=\left(a_{\pi(1)},\ldots, a_{\pi(n)}\right),
\end{array}\]
which is just a permutation of the entries of $A$.

\begin{definition}\label{24.01.15}
Let $C$ be a linear code. For an element $\pi$ of the symmetric group $S_n$, we define
\[\pi(C) := \{\pi(c) : c \in C \}.\]
\end{definition}

The {\it permutation group} of $C$ is the subgroup of the symmetric group $S_n$ defined by
\begin{equation*}
\Per(C) := \left\lbrace\pi\in S_n: \pi(C) = C \right\rbrace.
\end{equation*}

The permutation group tells us which coordinates of every element $c\in C$ we can permute and still get an element of the code $C$. If $G$ is the generator matrix of a code $C$, the permutation group asks for the columns we can permute in $G$ and still get a generator matrix of the code $C$.

The following is a classical result on the permutation group of a code that is relevant in the rest of the manuscript and can be proven directly from the definition.
\begin{lemma}\label{23.03.19}
For a linear code $C$, $\Per(C)=\Per(C^\perp)$.
\end{lemma}
\rmv{
\begin{proof}
($\subseteq$) Let $\pi$ be an element in $\Per(C)$ and $w$ an element in $C^\perp$. We aim to prove that $\pi(w) \in C^\perp$, which means that $c \cdot  \pi(w) = 0$ for any $c \in C$.

Take $c\in C$. As $\pi(C) = C$, there is $c^\prime \in C$ such that $\pi(c^\prime) =  c$. So, we have
\[ c \cdot  \pi(w) = \pi(c^\prime) \cdot \pi(w) = c^\prime \cdot w = 0.\]
Thus, $\pi(w) \in C^\perp$, meaning $\pi \in \Per(C^\perp)$.

($\supseteq$) By the previous paragraph, we have
\[\Per(C^\perp) \subseteq \Per((C^\perp)^\perp) = \Per(C),\]
which completes the proof.
\end{proof}}

The polynomial ring over $\F_q$ is denoted by $\F_q[x]$. Given $k \in \Z^{+}$, $\F_{q}[x]_{<k}$ denotes the set of polynomials of degree less than $k$. An element $f \in \F_{q}[x]$ defines the {\it evaluation map}
\[\begin{array}{lccc}
& \F_{q}^n & \to & \F_{q}^n \\
& A=(a_1,\ldots, a_n) & \mapsto & f(A):=\left(f(a_1),\ldots, f(a_n)\right).
\end{array}\]

For the remainder of the manuscript, $A = (a_1,\ldots, a_n)$ represents an element of $\F_q^n$ such that all the entries differ.
\begin{definition}
The {\it Reed-Solomon} (RS) code with evaluation set $A$ is defined by
\[ \RS(A,k) := \left\{ f(A) : f \in \F_{q}[x]_{<k} \right\}.\]
\end{definition}
Reed-Solomon codes $\RS(A,k)$ are $[n, k, n-k+1]$ codes over $\F_{q}$ with $n \leq q$, which means they are maximum distance separable.

\begin{remark}\label{24.01.10}
Let $f$ and $g$ be elements in $\mathbb{F}_{q}[x]_{<n}$. If $f(A)=g(A)$, then $f=g$. This is because if $f(a_i)=g(a_i)$ for $i=1,\ldots,n$, then $f-g$ is a polynomial of degree less than $n$ that has $n$ roots. So, $f-g$ is the zero polynomial, meaning $f=g$.
\end{remark}

A set of {\it indicator functions} for $\{a_1,\ldots, a_n \} \subseteq \Fq$ is a set $\{L_1,\ldots,L_n\}$ of $n$ polynomials in $\Fq[x]$ such that
\begin{equation*}
L_{i}(a_j)=\left\lbrace\begin{array}{cc}
     1&\text{if } i=j \\
     0&\text{if } i\neq j.
\end{array}\right.
\end{equation*}
These functions are well known in the theory of Lagrange interpolation~\cite{Gathen_Gerhard_2013}. In~\cite{LSV}, the authors used them to describe the dual of any evaluation code. There is a trivial way to construct a set of indicator functions where the degree of each function is less than $n$ (these are called {\it standard indicator functions}~\cite{LSV}). Indeed, given $\lbrace a_{1},\ldots,a_{n}\rbrace \subseteq \Fq$, define
%the set $\lbrace L_{1},\ldots, L_{n}\rbrace\subseteq\mathbb{F}_{q}[x]$, where
$$\displaystyle L_i(x):=\frac{\prod_{j\neq i}(x-a_{j})}{\prod_{j\neq i}(a_{i}-a_{j})}.$$
Note that $L_i(x) \in \Fq[x]_{<n}$ for $i=1,\ldots,n$.

%%%%%%%%%%%%%%%%%%%%%%%%%%%%%%%%%%%%%%%%%%%%%%%%%%%%%%%%%%%%%%%%%%%%%%%%%%%%%%%%%%%%%%%%%%%%%%%%%%%%%%%%%%%%%%%%%%%%%%%%%%%%%%%%%%%%%%%%%%%%%%%%%%%%%%%%%%%%%%%%%%%%%%%%%%%%%%
\section{Affine permutations}\label{affperm}
In this section, we prove that when $n<q$, the symmetric group $S_n$ can be viewed as a subset of $\Fq[x]_{<n}$. Then, we define affine permutations, which are associated with polynomials of degree $1$. The concept of affine permutations is essential to describe the permutation group of Reed-Solomon codes. From now on, we will assume that $n>1$.

We say that a polynomial $p$ in $\Fq[x]$ {\it permutes} $A$ if the function
\[\{a_1,\ldots,a_n\} \to \{a_1,\ldots,a_n\}, \qquad a \mapsto p(a),\] is a bijection.

The polynomial ring $\Fq[x]$ is not a group under the composition because not every element has an inverse. As a consequence of the standard indicator functions, we get a group when we restrict to the set
\[\left\{p \in \F_{q}[x]_{<n} : p \text{ permutes } A\right\}.\]
We must be careful with the operation because the composition of two polynomials of degree less than $n$ may have a degree larger than $n$.

Take $p_1, p_2 \in \left\{p \in \F_{q}[x]_{<n} : p \text{ permutes } A\right\}$. By the division algorithm, there exist $q(x)$ and $r(x)$ in $\F_{q}[x]$ such that for the usual composition of polynomials, we have
\[(p_1 \circ p_2)(x) = q(x)\prod_{i=1}^n(x-a_i) + r(x),\]
where $\deg(r) < n$. We define the composition of $p_1$ with $p_2$ modulo $A$ as
\[p_1 \underset{A}{\circ} p_2 := r.\]

Since $p_1$ and $p_2$ are permutations of $A$, their composition $p_1 \circ p_2$ also permutes $A$. Therefore, the composition $p_1 \underset{A}{\circ} p_2$ permutes $A$ as well, because 
$\left(p_1 \underset{A}{\circ} p_2\right)(a_i) = (p_1 \circ p_2)(a_i)$ for all $i=1,\ldots,n$. Moreover, since $\deg(p_1 \underset{A}{\circ} p_2) = \deg(r) < n$, it follows that $p_1 \underset{A}{\circ} p_2 \in \left\{p \in \F_{q}[x]_{<n} : p \text{ permutes } A\right\}$.

The following result follows directly from the previous lines.
\begin{lemma}
The set
\[\left\{p \in \F_{q}[x]_{<n} : p \text{ permutes } A\right\}\]
is a group under the composition modulo $A$.
\end{lemma}

The following result shows that every element $\pi \in S_n$ can be interpreted as a polynomial.
\begin{theorem}\label{25.01.15}
We have that
\begin{align*}S_n &\cong \left\{p \in \F_{q}[x]_{<n} : p \text{ permutes } A\right\} \\ \pi &\mapsto p_\pi\end{align*}
where each set is equipped with its respective group operations. Even more, we have $\pi(A) = p_\pi(A)$.
\end{theorem}
\begin{proof}
%Define $S:=\left\{p \in \F_{q}[x]_{<n} : p \text{ permutes } A\right\}$. Denote the map from $S_n$ to $S$ given in Lemma~\ref{24.01.11} by $\varphi$.

Let $\pi$ be an element of $S_n$ and $\{L_1, \ldots, L_n \}$ the set of standard indicator functions of $\{a_1,\ldots, a_n\}$. Define the polynomial
$p_{\pi} := \sum_{i=1}^{n}a_{\pi(i)}L_{i},$
which has degree less than $n$ and satisfies $p_{\pi}(a_i)= a_{\pi(i)}$. Observe that $p_\pi$ permutes $A$. If $p \in \Fq[x]$ permutes $A$, $\{a_1,\ldots,a_n\} = \{p(a_1),\ldots,p(a_n) \}$. So, there is a permutation $\pi_p$ of the set $\{1,\ldots,n \}$ such that $p(a_i) = a_{\pi_p(i)}$ for $i=1,\ldots,n$. The compositions of the two maps
$\pi \mapsto p_\pi$ and $p \mapsto \pi_p$ are the identity, so we have a set-theoretic bijection between the two sets.

Take $\pi_1$ and $\pi_2$ in $S_n$. Let $p_{\pi_1 \circ \pi_2}$, $p_{\pi_1}$, and $p_{\pi_2}$ be the associated polynomials to $\pi_1 \circ \pi_2$, $\pi_1$, and $\pi_2$, respectively by the bijection above. For $i=1,\ldots,n$, we have
\begin{align*}
\left(p_{\pi_1} \underset{A}{\circ} p_{\pi_2}\right)(a_i)
 = (p_{\pi_1} \circ p_{\pi_2})(a_i)
 %= p_{\pi_1}(a_{\pi_2(i)})
%& = a_{\pi_1(\pi_2(i))}\\
 %= a_{(\pi_1 \circ \pi_2)(i)}
 = p_{\pi_1 \circ \pi_2}(a_i),
\end{align*}
which means $p_{\pi_1} \underset{A}{\circ} p_{\pi_2} = p_{\pi_1 \circ \pi_2}$ by Remark~\ref{24.01.10}.

Finally, we have that
\begin{eqnarray*}
\pi(A) &=& \left(a_{\pi(1)}, \dots, a_{\pi(n)} \right)\\
&=& (p_\pi(a_1),\ldots, p_\pi(a_n)) = p_{\pi}(A).
\end{eqnarray*}
This completes the proof.
\end{proof}

\begin{remark}
    The previous proof is a particular case of the fact that any function $f$ in one variable over $\mathbb{F}_q$ corresponds to a polynomial $p$ in the sense that $f(a)=p(a)$ for any $a\in\mathbb{F}_q$.
\end{remark}

\begin{remark}\label{24.01.12}
Note that the image of the map in Theorem~\ref{25.01.15} lies in 
$\left\{p \in \F_{q}[x]_{\geq 1} : p \text{ permutes } A\right\}.$
This is because any constant polynomial does not permute $A$.
\end{remark}

%%%%%%%%%%%%%%%%%%%%%%%%%%%%%%%%%%%%%%%%%%%%%%%%%%%%%%%%%%%%%%%%%%%%%%%%%%%%%%%%%%%%%%%%%%%%%%%%%%%%%%%%%%%%%%%%%%%%%%%%%%%%%%%%%%%%%%%%%%%%%%%%%%%%%%%%%%%%%%%%%%%%%%%%%%%%%%
The set of polynomials of a fixed degree $k$ is denoted by $\F_{q}[x]_{=k}$. We are now ready to define the affine permutations of $A$.
\begin{definition}
An {\it affine permutation of $A$} is any permutation $\pi$ such that $p_{\pi}$ (from Theorem~\ref{25.01.15}) is in the set
\[\left\{p \in \F_{q}[x]_{=1} : p \text{ permutes } A\right\}.\]
By the isomorphism of Theorem~\ref{25.01.15}, any polynomial from the previous set is also known as an affine permutation of $A$.
\end{definition}

%%%%%%%%%%%%%%%%%%%%%%%%%%%%%%%%%%%%%%%%%%%%%%%%%%%%%%%%%%%%%%%%%%%%%%%%%%%%%%%%%%%%%%%%%%%%%%%%%%%%%%%%%%%%%%%%%%%%%%%%%%%%%%%%%%%%%%%%%%%%%%%%%%%%%%%%%%%%%%%%%%%%%%%%%%%%%%
\section{Permutation group of Reed-Solomon codes}\label{perRS}
In this section, we prove that the permutation group of a Reed-Solomon code is usually given by the affine permutations of the evaluation set. To be more precise, we give an elementary proof of the following facts.
\begin{itemize}
\item For $k=1$ and $k=n$,
\[\Per(\RS(A,k)) = S_n \cong \left\{p \in \F_{q}[x]_{<n} : p \text{ permutes } A\right\}.\]
%is given by all the elements in $\Fq[x]_{<n}$ that permute $A$.
\item For $k=2,\ldots,n-2$,
\[\Per(\RS(A,k))\cong \left\{p \in \F_{q}[x]_{=1} : p \text{ permutes } A\right\}.\]
%is given by all the polynomials of degree one that permute $A$.
%In other words, $\pi \in \Per(\RS(A,k))$ if and only if $\pi$ is an affine permutation of $A$.
\item For $k=n-1$,
\[ \Per(\RS(A,k)) \supset \left\{p \in \F_{q}[x]_{=1} : p \text{ permutes } A\right\}.\] The equality depends on $A$.
\end{itemize}
As a consequence, we obtain the following classical results
\begin{itemize}
\item For $k=2,\ldots,n-2$,
\[\Per(\RS(\F_q,k))\cong \F_{q}[x]_{=1}.\]
%is given by all the polynomials of degree one.
%In other words, $\pi \in \Per(\RS(\F_q,k))$ if and only if $\deg(p_\pi)=1$.  
\item For $k=2,\ldots,n-2$,
\[\Per(\RS(\F_q^*,k))\cong \left\{p \in \F_{q}[x]_{=1} : p(0) = 0 \right\}.\]
%is given by all the polynomials of degree one with zero constant.
\end{itemize}
\rmv{every affine permutation is an element of the permutation group of a Reed-Solomon code. Then we show that if $k \leq n/2$, then every element of the permutation group of a Reed-Solomon code $\RS(A,k)$ is an affine permutation. We prove that for any $1 < k < n$, every element of the permutation group $\Per(\RS(\F_q,k))$ is an affine transformation of the form $T(x)=ax+b$ and the automorphism group of the Reed-Solomon code $\RS(\F_q^*,k)$ is given by all linear transformations of the form $T(x)=ax$.}

Let $f$ be an element in $\F_{q}[x]$ and $\pi$ in $S_n$. In Section~\ref{preli}, we saw that $f$ and $\pi$ define maps from $\F_q^n$ to $\F_q^n$. The following result shows that these maps commute.
\begin{lemma}\label{24.01.16}
Let $f$ be an element in $\F_{q}[x]$ and $\pi$ in $S_n$. Using the maps defined in Section~\ref{preli}, the following diagram commutes.
\[\begin{tikzcd}
\F_{q}^n \arrow{r}{f} \arrow[swap]{d}{\pi} & \F_{q}^n \arrow{d}{\pi} \\
\F_{q}^n \arrow{r}{f} & \F_{q}^n
\end{tikzcd}
\]
\end{lemma}
\begin{proof}
For $A=(a_1,\ldots, a_n) \in \F_{q}^n$, we have that
\begin{eqnarray*}
(\pi \circ f)(A) 
&=&\pi (f(A))\\
&=& \pi\left(f(a_1), \dots, f(a_n) \right)\\
&=&\left(f(a_{\pi(1)}), \dots, f(a_{\pi(n)}) \right)\\
&=& f(\pi (A))\\
&=& (f\circ\pi)(A),
\end{eqnarray*}    
which completes the proof.
\end{proof}

We now show that the permutation group of any Reed-Solomon code $\RS(A,k)$ always contains all the affine permutations of $A$.
\begin{proposition}\label{24.01.13}
If $\pi \in S_n$ is an affine permutation of $A$, then $\pi \in \Per\left( \RS(A,k) \right)$.
\end{proposition}
\begin{proof}
Let $p_\pi$ be the image of $\pi$ under the isomorphism of Theorem~\ref{25.01.15}.

Let $f(A) = (f(a_1),\ldots,f(a_n))$ be an arbitrary element in $\RS(A,k)$, where $f \in \Fq[x]_{<k}$. We have
\begin{align*}
\pi(f(a_1),\ldots,f(a_n))
&= (\pi \circ f)(A)\\
&= (f \circ \pi)(A) \\
&= (f \circ p_\pi)(A),
\end{align*}
where the second equality holds by Lemma~\ref{24.01.16}, and the third equality holds by the fact that $\pi(A) = p_\pi(A)$. Moreover, $\deg(p_\pi) = 1$ since $\pi$ is an affine permutation of $A$. Thus, $\deg (f \circ p_\pi) = \deg (f) < k$, which means that $(f \circ p_\pi)(A)$ is in fact an element of $\RS(A,k)$. Thus, $\pi \in \Per\left( \RS(A,k) \right)$.
\end{proof}
%By an arbitrary polynomial in $S_n$, we mean the image of an arbitrary permutation in $S_n$ under the isomorphism of Theorem~\ref{25.01.15}.
In the following result, by $p^i$, we mean the power of a polynomial with respect to the usual product.
\begin{proposition}\label{24.05.03}
Assume $0 < k < n$ and define $k^\prime:= n-k$. Let $\pi$ be an element in $\Per(\RS(A,k))$. The polynomial $p_\pi$ associated with $\pi$ under the isomorphism of Theorem~\ref{25.01.15} satisfies the following.
\begin{itemize}
\item [(i)] For $i<k$, there exists $ f_i \in \Fq[x]_{<k}$ such that
\[p_{\pi}^i(A) = f_i(A).\]
\item [(ii)] For $i<k^\prime$, there exists $g_i \in \Fq[x]_{<k^\prime}$ such that 
\[((g\circ p_\pi) \cdot p_\pi^i)(A)= (g \cdot g_i)(A),\] 
for some $g(x) \in \mathbb{F}_q[x]_{<n}$ such that $g(A)$ has no zero entries.
\end{itemize}
\end{proposition}
\begin{proof}
Let $\pi$ be an element $\Per\left( \RS(A,k) \right)$ and $p_{\pi}$ their corresponding polynomial under the isomorphism of Theorem~\ref{25.01.15}.

(i) Take $0 \leq i < k$. We have
\begin{align*}
%\left(p_{\pi}^i(a_1),\ldots,p_{\pi}^i(a_n)\right)
p_{\pi}^i(A)
&= (x^i \circ p_{\pi})(A)\\
&= (x^i \circ \pi)(A)\\
&= (\pi \circ x^i)(A)\\
&= \pi \left(x^i(A)\right).
\end{align*}
As $i<k$, the element $x^i(A) = \left(x^i(a_1),\ldots,x^i(a_n) \right)$ belongs to $\RS(A,k)$. We also have that $\pi \in \Per\left( \RS(A,k) \right)$, so $ p_{\pi}^i(A) = \pi \left( x^i(A)\right) \in \RS(A,k)$. Then, there is an element $f_i \in \Fq[x]_{<k}$ such that
\[
p_{\pi}^i(A) = f_i(A).
\]

(ii) There exists a polynomial $g(x) \in \mathbb{F}_q[x]_{<n}$ such that $g(A)$ has no zero entries and the dual of the Reed-Solomon code $\RS(A,k)$ is given by
\begin{align*}
    g(A) \star \RS(A,k^\prime) & = \{g(A) \star f(A) : \deg(f)< k^\prime\}\\
    & = \{(g \cdot f)(A) : \deg(f)< k^\prime\},
\end{align*}
where $\star$ represents the component-wise product and $k^\prime = n-k \geq 1$. Take $0 \leq i < k^\prime$. As $\pi$ is an element in $\Per\left( \RS(A,k) \right)$, by Lemma~\ref{23.03.19}, $\pi$ belongs also to $\Per\left( g(A)\star\RS(A,k^\prime) \right)$. Then, we obtain
\begin{align*}
((g\circ p_\pi) \cdot p_\pi^i)(A)
&= ((g \circ p_\pi) \cdot (x^i \circ p_\pi))(A)\\
&= ((g \cdot x^i) \circ p_{\pi})(A)\\
&= ((g \cdot x^i) \circ \pi)(A)\\
&= (\pi \circ (g \cdot x^i))(A)\\
&= \pi \left( \left( g\cdot x^i\right) (A) \right).
\end{align*}
As $i < k^\prime$, $\left( g\cdot x^i\right) (A) \in g(A) \star \RS(A,k^\prime)$. We also have that $\pi \in \Per\left( g(A) \star \RS(A,k^\prime) \right)$, so
\[((g\circ p_\pi) \cdot p_\pi^i)(A) = \pi \left( \left( g\cdot x^i\right) (A) \right) \in g(A) \star \RS(A,k^\prime).\]
Thus, there is an element $g_i \in \Fq[x]_{<k^\prime}$ such that
\[
((g\circ p_\pi) \cdot p_\pi^i)(A)=(g \cdot g_i)(A), \qquad \text{ for } i < k^\prime,
\]
which completes the proof.
\end{proof}
The following result bounds the degree of polynomials associated with the permutations of Reed-Solomon codes.
\begin{theorem}\label{24.01.17}
Assume $1 < k < n-1$. Let $\pi$ be an element in $\Per(\RS(A,k))$. For the associated polynomial $p_\pi$ under the isomorphism of Theorem~\ref{25.01.15}, we have that
\[\deg(p_\pi)<\min \{k,n-k\}.\]
\end{theorem}
\begin{proof}
Define $k^\prime:= n-k$. By the definition of $p_\pi$ in Theorem~\ref{25.01.15}, $\deg(p_{\pi}) < n$. By Proposition~\ref{24.05.03} (i), $p_{\pi}(A) = f_1(A)$, with $\deg(f_1) < k$. Thus, by Remark~\ref{24.01.10}, $p_{\pi}=f_1$ and
\begin{eqnarray}\label{24.05.04}
\deg(p_{\pi}) = \deg(f_1) < k.
\end{eqnarray}
By Proposition~\ref{24.05.03} (ii), the values $i=0$ and $i=1$ imply
$$(g\circ p_\pi)(A)=(g \cdot g_0)(A)$$
and $$\left((g\circ p_\pi)\cdot p_\pi\right)(A)=(g \cdot g_1)(A).$$
The previous expressions give rise to the equations
\[(g \cdot g_0 \cdot p_\pi)(A)
= \left((g\circ p_\pi)\cdot p_\pi\right)(A)
= (g \cdot g_1)(A),\]
which imply that $g(a_j) (g_0 \cdot p_\pi - g_1)(a_j) = 0$ for $j=1,\ldots, n$. As $g(a_j) \neq 0$ for $j=1,\ldots,n$, then
\[(g_0\cdot p_\pi)(A) = g_1(A).\]
Proposition~\ref{24.05.03} (ii) and Eq.~(\ref{24.05.04}) imply that $\deg(g_0) < k^\prime,$ $\deg(g_1)<k^\prime$, and $\deg(p_{\pi})<k$. Thus, $\deg (g_0\cdot p_\pi) < k^\prime + k = n$. Therefore, by Remark~\ref{24.01.10},
\[g_0 \cdot p_\pi = g_1 \in \mathbb{F}_q[x]_{<k^\prime}.\]
Thus $\deg(p_\pi) < k^\prime$, which completes the proof together with Eq.~(\ref{24.05.04}).
\end{proof}

\begin{proposition}\label{24.01.14}
Take $1 < k < n-1$. If $\pi \in \Per(\RS(A,k))$, then $\pi$ is an affine permutation of $A$.
\end{proposition}
\begin{proof}
Define $k^\prime:= \min\{k,n-k\}>1$. If $k \leq n/2$, then $k^\prime = k$. Otherwise, $k > n/2$, meaning $k^\prime = n-k < n-n/2=n/2$. In any case,
\[1 < k^\prime \leq n/2 \quad \text{ and } \quad k^\prime + k \leq n.\]
We will use these two facts in the rest of the proof.

Let $p_\pi$ be the image of $\pi$ under the isomorphism of Theorem~\ref{25.01.15}. By Proposition~\ref{24.05.03}~(i), there exist $f_i \in \Fq[x]_{<k}$ such that
\[p_{\pi}^i(A) = f_i(A)\]
for $i=1,\ldots,k-1$. 
\begin{itemize}
\item[-] By Theorem~\ref{24.01.17}, $\deg(p_{\pi})<k^\prime$.
\item[-] We have
\[
\deg(p_{\pi}^2) = \deg(p_{\pi}) + \deg(p_{\pi})
< k^\prime + k^\prime
\leq n.
\]
%\begin{alignat*}{6} % 4 is the number of equation columns
%\deg(p_{\pi}^2) &= &\deg(p_{\pi}) &+ &\deg(p_{\pi})&\\
%&< &k^\prime &+ &k^\prime &\leq n.
%\end{alignat*}
By Remark~\ref{24.01.10}, we get $p_{\pi}^2 = f_2 \in \Fq[x]_{<k}$.
\item[-] We have
\[
\deg(p_{\pi}^3) = \deg(p_{\pi}) + \deg(p^2_{\pi})
< k^\prime + k
\leq n.
\]
We get $p_{\pi}^3 = f_3 \in \Fq[x]_{<k}$ by Remark~\ref{24.01.10}.
\item[-] In general, using induction for $i=4,\ldots k-1$, we obtain
\[
\deg(p_{\pi}^i) = \deg(p_{\pi}) + \deg(p^{i-1}_{\pi})
< k^\prime + k
\leq n.
\]
Thus, By Remark~\ref{24.01.10}, we get $p_{\pi}^i = f_i \in \Fq[x]_{<k}$ for $i=1,\ldots k-1$.
\end{itemize}
Taking $i=k-1$, we obtain $p_{\pi}^{k-1} \in \Fq[x]_{<k}$, which means that $\deg(p_{\pi})<k/(k-1) < 2$. By Remark~\ref{24.01.12}, we get the result.
\end{proof}

The following result is an elementary proof that the permutation group of a Reed-Solomon code is given by affine permutations when $1 < k < n-1$.
\begin{theorem}\label{24.01.19}
For $1 < k < n-1$, we have 
\[\Per(\RS(A,k))\cong \left\{p \in \F_{q}[x]_{=1} : p \text{ permutes } A\right\}.\]
In other words, the permutation group of a Reed-Solomon code is given by the affine permutations of $A$.
\end{theorem}
\begin{proof}
If $\pi$ is an affine permutation, then $\pi \in \Per(\RS(A,k))$ by Proposition~\ref{24.01.13}. In addition, $\Per(\RS(A,k))$ contains only affine permutations by Propositions~\ref{24.01.14}. Thus, we get the result.
\end{proof}

As a consequence of the previous result, we obtain that the permutation group of a Reed-Solomon code with evaluation set $A=\Fq$ is given by all affine polynomials in $\F_q[x]$.
\begin{corollary}
For $1 < k < n$, we have 
\[\Per(\RS(\F_q,k))\cong \F_{q}[x]_{=1}.\]
%In other words, the permutation group of a Reed-Solomon code is given by the affine polynomials in $\F_q[x]$
\end{corollary}
\begin{proof}
As every element in $\F_{q}[x]_{=1}$ permutes $\F_{q}$, we get that 
\[\left\{p \in \F_{q}[x]_{=1} : p \text{ permutes } \Fq\right\} = \F_{q}[x]_{=1}.\]
Thus, the result follows from Theorem~\ref{24.01.19}.
\end{proof}

We now prove that the permutation group of the Reed-Solomon code when $A=\F_q^*$ is defined by all affine permutations with a zero constant.

\begin{corollary}
For $1 < k < n$, we have 
\[\Per(\RS(\F_q^*,k))\cong \left\{p \in \F_{q}[x]_{=1} : p(0) = 0 \right\}.\]
\end{corollary}
\begin{proof}
An element $p$ in $\F_{q}[x]_{=1}$ permutes $\F_q^*$ if and only if $p = ax+0$. Thus, we get that 
\[\left\{p \in \F_{q}[x]_{=1} : p \text{ permutes } \F_q^*\right\} = \left\{p \in \F_{q}[x]_{=1} : p(0) = 0 \right\}.\]
The result follows from Theorem~\ref{24.01.19}.
\end{proof}

As the following example shows, Theorem~\ref{24.01.19} may not be true when $ k = n-1$.
\begin{example}
We consider an example where the dimension $k$ does not satisfy the condition $1 < k < n-1$. Take $A = \left(0,1,4,6\right) \in \F_{13}^4$ and $k=3 = 4 - 1$. A generator matrix of the Reed-Solomon code $\RS(A,3)$ is given by
\[G =
\begin{bmatrix}
1 & 1 & 1 & 1 \\
0 & 1 & 4 & 6 \\
0 & 1 & 3 & 10
\end{bmatrix}.\]
It can be verified with SAGE~\cite{sage}, Magma~\cite{magma}, or {\it Macaulay2}~\cite{Mac2} using the Coding Theory Package~\cite{cod_package} that the permutation group $\Per(\RS(A,3))$ is isomorphic to $S_3$. However, only three affine polynomials $x, 3x+1, 9x+4$ preserve $A$. These polynomials correspond to the elements $(),(123),(132) \in S_3$ which form a proper subgroup. Therefore, the affine group does not generate the permutation automorphism group. \footnote{The code to reproduce the example can be found on \texttt{https://github.com/junbolau/permutation{\_}group{\_}RS{\_}code}}
\end{example}

\begin{remark}\rm
It is important to highlight that the permutation group is just one of the various groups of isometries that can be considered when working with a linear code $C$. The other two important groups are the following.
\begin{itemize}
\item The automorphism group $Aut(C)$, which includes transformations of the form $(c_1,\ldots,c_n)\mapsto (v_1 c_{\pi(1)},\ldots,v_nc_{\pi(n)})$, where $v\in(\mathbb{F}_q^\ast)^n$ and $\pi\in S_n$.
\item The semilinear automorphism group $SAut(C)$, which includes transformations of the form $(c_1,\ldots,c_n)\mapsto (v_1\tau( c_{\pi(1)}),\ldots,v_n\tau(c_{\pi(n)}))$, where $v$ and $\pi$ are as before and $\tau$ is an autormophism of $\mathbb{F}_q$.
\end{itemize}
\end{remark}
For a Reed-Solomon code $C=\RS(A,k)$, the following example shows that the semilinear automorphism group $SAut(C)$ does not only depend on $A$.
\begin{example}
Let $q=9$ and $\alpha$ be a primitive element of $\mathbb{F}_9$ such that $\alpha^3=-\alpha+1$. Let $A=(0,1,2,\alpha^2,\alpha^6)$ and $C=\RS(A,4)$. We have that $\tau(y)=y^3$ is an automorphism of the field. Since $\tau(x^2)=x^6$ and $x^2(A)=x^6(A)$ and $x^9(A)=x(A)$, then $C$ is fixed under the action of $\tau$ and thus $\mathrm{Per}(C)\subsetneq SAut(C)$. 

On the other hand, the Reed-Solomon code $D=RS(A,3)$ is not fixed by $\tau$ since $\tau(x)=x^3$. 
\end{example}

\section{Conclusions}\label{conclusion}
In this work, we proved that the permutation group of a Reed-Solomon code is given by polynomials of degree one that leave the set of evaluation points invariant. Our results showed an elementary proof of the well-known cases of the permutation group of the Reed-Solomon code when the set of evaluation points is the whole finite field or the multiplicative group. As future work, we are interested in determining the set of semilinear isometries of Reed-Solomon codes. We also aim to extend the presented technique to compute the automorphism group of generalized Reed-Solomon and Goppa codes.

\bibliographystyle{IEEEtran}
\bibliography{biblio}

\end{document}